\documentclass{llncs}

\usepackage[english]{babel}
\usepackage[pdftex, colorlinks, linkcolor=black, citecolor=black]{hyperref}
\usepackage{amsfonts, amsmath, xspace}
\usepackage[linesnumbered,ruled,lined]{algorithm2e}
\usepackage{fancyvrb}
\usepackage{graphicx}

\newcommand{\etal}{\emph{et al.}\xspace}
\newcommand{\comp}{\textsf{Companion}}
\newcommand{\F}{\mathbb{F}}
\newcommand{\C}{\Gamma}

\renewcommand{\a}{\alpha}
\renewcommand{\b}{\beta}
\newcommand{\z}{z}

\newcommand{\Min}{\mathrm{Min}}
\newcommand{\ord}{\mathrm{ord}}
\newcommand{\lcm}{\mathrm{lcm}}

\title{Direct Construction of Recursive MDS Diffusion Layers using Shortened BCH Codes}
\author{Daniel Augot\inst{1}\and Matthieu Finiasz\inst{2}}
\institute{INRIA - LIX UMR 7161 X-CNRS \and CryptoExperts}

\begin{document}

\maketitle

\begin{abstract}
  MDS matrices allow to build optimal linear diffusion layers in block
  ciphers. However, MDS matrices cannot be sparse and usually have a
  large description, inducing costly software/hardware
  implementations. Recursive MDS matrices allow to solve this problem
  by focusing on MDS matrices that can be computed as a power of a
  simple companion matrix, thus having a compact description suitable even
  for constrained environments. However, up to now, finding recursive
  MDS matrices required to perform an exhaustive search on families of
  companion matrices, thus limiting the size of MDS matrices one could
  look for. In this article we propose a new \emph{direct}
  construction based on shortened BCH codes, allowing to efficiently
  construct such matrices for whatever parameters. Unfortunately, not
  all recursive MDS matrices can be obtained from BCH codes, and our
  algorithm is not always guaranteed to find the best matrices for a
  given set of parameters.
\end{abstract}

\section{Introduction}

Diffusion layers are a central part of most block cipher
constructions. There are many options when designing a diffusion
layer, but linear diffusion is usually a good choice as it can be
efficient and is easy to analyze. The quality of a linear diffusion
layer is connected to its \emph{branch number}~\cite{daemen-thesis}: the minimum over all
possible nonzero inputs of the sum of the Hamming weights of the input and the
corresponding output of this diffusion layer. A high branch number
implies that changing a single bit of the input will change the output
a lot, which is exactly what one expects from a good diffusion layer.
Before going into more details on how to build linear diffusion with a
high branch number, let us recall some elements of coding theory.

\subsubsection{Linear diffusion and coding theory.}

A linear code $\C$ of dimension $k$ and length $n$ over $\F_q$
(denoted as an $[n,k]_q$ code) is a vectorial subspace of dimension
$k$ of $(\F_q)^n$. Elements of $\C$ are called code words. The minimal
distance $d$ of a code is the minimum over all nonzero code words $c\in\C$ of
the Hamming weight of $c$. A $[n,k]_q$ code of minimal distance $d$
will be denoted as an $[n,k,d]_q$ code. A generator matrix $G$ of a
code is any $k\times n$ matrix over $\F_q$ formed by a basis of the
vectorial subspace $\C$. We say a generator matrix is in systematic
form when it contains (usually on the left-most positions) the
$k\times k$ identity matrix $I_k$. The non-systematic part (or redundancy
part) of $G$ is the $k\times (n-k)$ matrix next to this identity
matrix.

Now, suppose a linear diffusion layer of a block cipher is defined by an invertible
matrix $M$ of size $k\times k$ over $\F_q$, so that an input
$x\in(\F_q)^k$ yields an output $y\in(\F_q)^k$ with $y = x\times
M$. Then, the $k\times 2k$ generator matrix $G_M$ having $M$ as its
non-systematic part (the matrix defined as the concatenation of the
$k\times k$ identity matrix $I_k$ and of $M$, as $G_M=[I_k\mid M]$)
generates a $[2k,k]_q$ code $\C_M$ whose minimal distance is exactly
the branch number of $M$. Indeed, a code word $c = x\times G_M$ in
$\C_M$ is the concatenation of an input $x$ to the diffusion layer and
the corresponding output $y = x\times M$. So the Hamming weight of
every code word is the sum of the Hamming weights of an input and its
output.

Optimal linear diffusion can thus be obtained by using codes with the
largest possible minimal distance, namely maximum distance separable
(MDS) codes. A $[n,k]_q$ code is called MDS if its minimal distance is
$d = n-k+1$. By extension, we will say that a matrix $M$ is MDS when
its concatenation with the identity matrix yields a generating matrix
$G_M$ of an MDS code $\C_M$. In the context of diffusion where $n=2k$
being MDS means that $d = k+1$: changing a single element in the input
of the diffusion layer will change all the elements in its output.

We also recall the MDS conjecture: if there exists an $[n,k]_q$ MDS
code, meaning an MDS code of length $n$ and dimension $k$
over $\F_q$, then $n\leq q+1$, except for particular
cases which are not relevant to our context. All along this article we
will assume that this conjecture holds~\cite{sloane}.

\paragraph{Note on Vector Representation.} In coding theory, vectors
are usually represented as rows (with $y = x\times M$), as we have
done for the moment. In cryptography, however, they are more often
represented as columns (with $y = M^T\times x$). Luckily, the
transposed of an MDS matrix is also MDS, so if $G_M$ defines an MDS
code, both $M$ and $M^T$ can be used as MDS diffusion matrices. In the
rest of the article we will use the column representation, which
people used to the AES and the MixColumns operation are more familiar
with: the diffusion layer defined by a matrix $M$ computes $y=M\times
x$. This way, the branch number of $M$ is the minimal distance of the
code generated by $G_{M^T} = [I_k\mid M^T]$. However, in order to
avoid matrix transpositions, we will rather check wether $G_M =
[I_k\mid M]$ generates an MDS code or not.

\subsubsection{Recursive MDS matrices.} MDS matrices offer optimal
linear diffusion, but in general, they do not allow for a very compact
description. Indeed, the non-systematic part $M$ of an MDS generator
matrix cannot contain any 0 element\footnote{If the non-systematic
  part $M$ of an MDS generator matrix contained a 0, then the line of
  $G_M$ containing this zero would have Hamming weight $\leq k$, which
  is in contradiction with the minimal distance of the code. More
  generally, for an MDS code $\C_M$, for any $i\leq k$ all the
  $i\times i$ minors of $M$ must be non-zero.}. These matrices can
never be sparse and applying such a matrix to its input requires a full
matrix multiplication for the diffusion. Several different techniques
have been studied to obtain \emph{simpler} MDS matrices, a well known
example being circulant matrices (or modifications of circulant matrices) as
used in the AES~\cite{aes} or FOX~\cite{fox}. Recently a new
construction has been proposed: the so-called \emph{recursive} MDS
matrices, that were for example used in Photon~\cite{photon} or
LED~\cite{LED}. These matrices have the property that they
can be expressed as a power of a companion
matrix $C$. For example, in Photon, using the same decimal
representation of elements of $\F_{256}$ as in~\cite{photon}:

$$\arraycolsep=2pt M = \begin{pmatrix}1 & 2 & 1 & 4\\
4 & 9 & 6 & 17\\
17 & 38 & 24 & 66\\
66 & 149 & 100 & 11\end{pmatrix} = C^4,
\ \textrm{with}\
\arraycolsep=4pt C = \begin{pmatrix}0 & 1 & 0 & 0\\
0 & 0 & 1 & 0\\
0 & 0 & 0 & 1\\
1 & 2 & 1 & 4\end{pmatrix} = \comp(1,2,1,4).$$

The advantage of such matrices is that they are particularly well
suited for lightweight implementations: the diffusion layer can be
implemented as a linear feedback shift register that is clocked 4
times (or more generally $k$ times), using a very small number of
gates in hardware implementations, or a very small amount of memory
for software. The inverse of the diffusion layer also benefits from a
similar structure, see Eq.~\eqref{eq.regular_inverse} for a particular
case.

\subsubsection{Outline.} In the next section, we will present previous
methods that have been used to find recursive MDS matrices. Then, in
Section~\ref{sec.BCH}, we will introduce BCH codes and shortened BCH
codes, show that they too can yield recursive MDS matrices, and give a
direct construction of such matrices. In Section~\ref{sec.algo} we
will then describe an algorithm to explore all BCH codes and the MDS
diffusion layers they yield for given parameters. We will conclude
with a few experimental results.

\section{Exhaustive Search for Recursive MDS Matrices}

Exhaustive search for recursive MDS matrices can be quite straightforward:
\begin{itemize}
\item pick some parameters: the matrix size $k$ and the field size $q = 2^s$,
\item loop through all companion matrices $C$ of size $k$ over $\F_q$,
\item for each $C$, computes its $k$-th power and check if it is MDS.
\end{itemize}
However, this technique is very expensive as there are many companion
matrices ($2^{ks}$, which could be $2^{128}$ for an 128-bit cipher)
and checking if a matrix is MDS is also expensive (the number of
minors to compute is exponential in $k$). Also, it does not specially
explore the most efficient matrices first. In the Photon example, the
matrix uses very sparse coefficients (the field elements represented
by 1, 2 and 4) to make the implementation of their operations on inputs even more
efficient. Exhaustive search should focus on such matrices.

Following this idea, Sajadieh \etal~\cite{Sajadieh:FSE2012} proposed
to split the search in two. Their companion matrices are symbolic
matrices $C(X)$ which have coefficients in the polynomial ring
$\F_q[X]$ where $X$ is an indeterminate, which will be substituted
later by some $\F_2$-linear operator $L$ of $\F_q$. Then their search
space is reduced to symbolic companion matrices $C(X)$ whose
coefficients are small degree polynomials in $X$ (small degree
polynomials will always yield a rather efficient matrix). Once $C(X)$
is raised to the power $k$, to get $D(X)=C(X)^k$, the matrix $D(X)$
will give an MDS matrix $D(L)$ when evaluated at a particular $L$, if
for all $i\leq k$, all its $i\times i$ minors evaluated at $L$ are
invertible matrices (non-zero is enough in a field, but now the
coefficients are $\F_2$-linear operators). Indeed, for a symbolic
matrix $D(X)$, the minors are polynomials in $X$, and their evaluation
at a particular linear operator $L$ needs to be invertible matrices.

This way, for each matrix $C(X)$ explored during the search, the
minors of all sizes of $D(X)=C(X)^k$ are computed: some matrices have
minors equal to the null polynomial and can never be made MDS when $X$
is substituted by a linear operator $L$, for the others this gives
(many) algebraic polynomials in $X$ which must not vanish when
evaluated at $L$, for the $k$-th power $D(L)$ to be MDS. Then, the
second phase of the search of Sajadieh \etal is to look for efficient
operators $L$ such that all the above minors are non zero when
evaluated at $L$. The advantage of this technique is that it finds
specially efficient recursive MDS matrices, but the computations of
the minors of symbolic matrices can be pretty heavy, because of the
growth of the degree of the intermediate polynomials involved. In the
case of Photon, the matrix could be found as $C = \comp(1, L, 1, L^2)$
where $L$ is the multiplication by the field element represented by 2.

Continuing this idea and focusing on hardware implementation, Wu,
Wang, and Wu~\cite{WuWangWu:SAC2012} were able to find recursive MDS
matrices using an impressively small number of XOR gates. They used a
technique similar to Sajadieh~\etal, first searching for symbolic
matrices with a list of polynomials having to be invertible when
evaluated in $L$, then finding an $\F_2$-linear operator $L$ using a
single XOR operation and with a minimal polynomial not among the list
of polynomials that have to be invertible.

Then, looking for larger recursive MDS matrices, Augot and
Finiasz~\cite{augot-finiasz-isit13} proposed to get rid of the
expensive symbolic computations involved in this technique by choosing
the minimal polynomial of $L$ \emph{before} the search of companion
matrices $C(X)$. Then, all computation can be done in a finite field (modulo
the chosen minimal polynomial of $L$), making them much faster. Of
course, assuming the MDS conjecture holds, the length of the code
cannot be larger than the size of the field plus one, so for an $L$
with irreducible minimal polynomial of degree $s$, the field is of
size $q=2^s$, and $k$ must verify $2k\leq 2^s+1$. Larger MDS matrices
will require an operator $L$ with a higher degree minimal
polynomial. Also, in the case where the bound given by the MDS
conjecture is almost met (when $k = 2^{s-1}$), Augot and Finiasz noted
that all companion matrices found had some kind of symmetry: if the
$k$-th power of $\comp(1,c_1,c_2,\dots,c_{k-1})$ is MDS, then $c_i =
c_{k-i}$ for all $1\leq i\leq\frac{k-1}2$.

\subsection{An Interesting Example}\label{sec.example}

One of the \emph{symmetric} MDS matrices found by Augot and Finiasz~\cite{augot-finiasz-isit13} for $k=8$ and $\F_q = \F_{16}$ is
\[C = \comp(1,
\a^3, \a^4, \a^{12}, \a^8, \a^{12}, \a^4, \a^3)
\] with $\a^4+\a+1 = 0$.
 As we will see later, there is a strong link
between companion matrices and the associated polynomial, here
\[
P_C(X) = 1 + \a^3 X + \a^4 X^2+ \a^{12} X^3+ \a^8 X^4+
\a^{12}X^5 + \a^4 X^6+ \a^3X^7+X^8.
\] In this example, this polynomial factors into terms of degree two:
\[
P_C(X)=(1+\a^2X+X^2)(1+\a^4X+X^2)(1+\a^8X+X^2)(1+\a^9X+X^2),
\] meaning that $P_C(X)$ is split in a degree-2 extension of
$\F_{16}$, the field $\F_{256}$.

If we now consider $P_C(X)$ in $\F_{256}[X]$, which we can, since
$\F_{16}$ is a subfield of $\F_{256}$, and look for its roots in
$\F_{256}$, we find that there are 8 roots in $\F_{256}$, which, for a
certain primitive $255$-th root of unity $\b\in\F_{256}$, are
\[[\b^5, \b^6, \b^7,\b^8,\b^9,\b^{10},\b^{11}, \b^{12}].\]
This indicates a strong connection with BCH codes that we will now
study.

\section{Cyclic Codes, BCH Codes, and Shortening}\label{sec.BCH}

Before jumping to BCH codes, we must first note a few things that are
true for any cyclic code and not only BCH codes. For more details on
the definition and properties of cyclic codes, the reader can refer
to~\cite{sloane}.

\subsection{A Systematic Representation of Cyclic Codes}

An $[n,k]_q$ code is said to be cyclic if a cyclic shift of any
element of the code remains in the code. For example, the code defined
by the following generator matrix $G$ over $\F_2$ is cyclic:
$$G = \arraycolsep=4pt \begin{pmatrix}
  1 & 0 & 1 & 1 & 0 & 0 & 0\\
  0 & 1 & 0 & 1 & 1 & 0 & 0\\
  0 & 0 & 1 & 0 & 1 & 1 & 0\\
  0 & 0 & 0 & 1 & 0 & 1 & 1\end{pmatrix}.$$

\noindent A cyclic shift to the right of the last line of $G$ gives
$(1,0,0,0,1,0,1)$ which is the sum of the first, third and last lines
of $G$, thus remains in the code: $G$ indeed generates a cyclic code.

Cyclic codes can also be defined in terms of polynomials:
$(1,0,1,1,0,0,0)$ corresponds to $1+X^2+X^3$ and a cyclic shift to the right
is a multiplication by $X$ modulo $X^n-1$. This way, cyclic codes can
be seen as ideals of $\F_q[X]/(X^n-1)$, meaning that each cyclic code
$\C$ can be defined by a generator polynomial $g(X)$ such that $\C=<g(X)>$
and $g(X)$ divides $X^n-1$. Then, the code defined by $g(X)$ has
dimension $k = n-\deg(g)$. In our example, $g(X) = 1+X^2+X^3$, which
divides $X^7-1$, and the code is indeed of dimension 4.

Any multiple of $g(X)$ is in the code, so for any polynomial $P(X)$ of
degree less than $n$, the polynomial $P(X) - (P(X)\bmod g(X))$ is in
the code. Using this property with $P(X) = X^i$ for $i\in[\deg(g),
n-1]$, we obtain an interesting systematic form for any cyclic code
generator matrix:

$$G = \arraycolsep=4pt \begin{pmatrix}
-X^3\bmod g(X)\smash{\hskip3pt\vrule width 1pt height 3mm depth 13.7mm\hskip-4pt}& 1 & 0 & 0 & 0\\
-X^4\bmod g(X)& 0 & 1 & 0 & 0\\
-X^5\bmod g(X)& 0 & 0 & 1 & 0\\
-X^6\bmod g(X)& 0 & 0 & 0 & 1\end{pmatrix} =  \begin{pmatrix}
1 & 0 & 1\smash{\hskip4pt\vrule width 1pt height 3mm depth 13.7mm\hskip-5pt}& 1 & 0 & 0 & 0\\
1 & 1 & 1 & 0 & 1 & 0 & 0\\
1 & 1 & 0 & 0 & 0 & 1 & 0\\
0 & 1 & 1 & 0 & 0 & 0 & 1\end{pmatrix}.$$

This form is exactly what we are looking for when searching for powers
of companion matrices. Indeed, if we associate the companion matrix $C
= \comp(c_0,\dots,c_{k-1})$ to the polynomial $g(X) =
X^k+c_{k-1}X^{k-1}+\dots+c_0$, then the successive powers of $C$ are (continuing with our example where $k=3$):
$$C = \begin{pmatrix}0\hskip5mm 1\hskip5mm  0\\
0\hskip5mm  0\hskip5mm  1 \\
-X^3\bmod g(X)\end{pmatrix}\hskip-.5mm,\ C^2 = \begin{pmatrix}0\hskip5mm  0\hskip5mm  1 \\
-X^3\bmod g(X)\\
-X^4\bmod g(X)\end{pmatrix}\hskip-.5mm,\ C^3 = \begin{pmatrix}-X^3\bmod g(X)\\
-X^4\bmod g(X)\\
-X^5\bmod g(X)\end{pmatrix}\hskip-.5mm.$$

\noindent To build recursive MDS matrices we thus simply need to build MDS cyclic codes with suitable parameters and
their corresponding $g(X)$.

Note that a multiplication by a companion matrix can also be expressed
in terms of LFSR. Initializing the LFSR of \figurename~\ref{fig.lfsr}
with a vector and clocking it once corresponds to the multiplication
of this vector by $C$. Clocking it $k$ times corresponds to the
multiplication by $M=C^k$. We will continue using the matrix
representation in the rest of the paper, but most results could also
be expressed in terms of LFSR.

\begin{figure}[t]
  \centering
  \includegraphics[width=5cm]{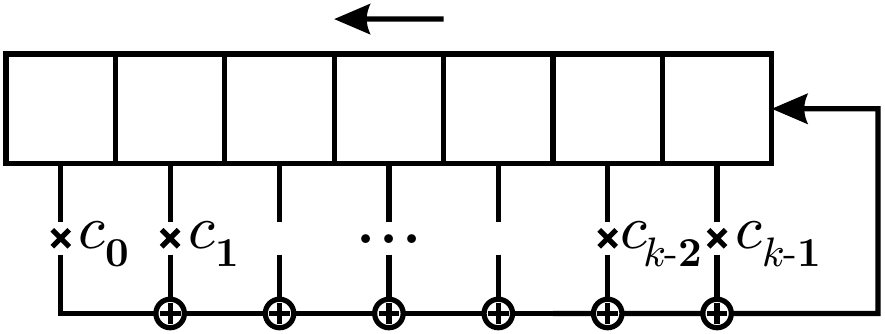}
  \caption{\label{fig.lfsr}An LFSR corresponding to the companion
    matrix $C$ of polynomial $g(X) = X^k+c_{k-1}X^{k-1}+...+c_0$. Clocking it $k$ times is equivalent to applying $C^k$ to its internal state.}
\end{figure}

\subsection{BCH Codes and Shortened BCH Codes}

In general, given a generator polynomial $g(X)$, computing the minimal
distance of the associated cyclic code is a hard problem. For
instance, the code generated by $g(X) = 1+X^2+X^3$ in the example of
the previous section has minimal distance 3, but even for such small
examples it is not necessarily immediate to find the minimum
distance. Nonetheless, lower bounds exist for some specific
constructions. This is the case for BCH codes, as described for
example in~\cite{sloane}.

\begin{definition}[BCH codes]
  A BCH code over $\F_q$ is defined using an element $\b$ in some
  extension $\F_{q^m}$ of $\F_q$. First, pick integers $\ell$ and $d$
  and take the $(d-1)$ consecutive powers $\b^\ell,
  \b^{\ell+1},\dots,\b^{\ell+d-2}$ of $\b$, then compute $g(X) =
  \lcm(\Min_{\F_q}(\b^\ell),\dots,\Min_{\F_q}(\b^{\ell+d-2}))$, where
  $\Min_{\F_q}(\b^\ell)$ is the minimal polynomial of $\b^\ell$ over
  $\F_q$.

  The cyclic code over $\F_q$ of length $\ord(\b)$ defined by $g(X)$
  is called a \emph{BCH code}, it has dimension $(\ord(\b)-\deg(g))$ and has
  minimal distance at least $d$. We write this as being an $[\ord(\b),
  \ord(\b)-\deg(g), \geq d]_q$ code.
\end{definition}

For such a BCH code to be MDS, $g(X)$ must have degree exactly $d-1$
(for a cyclic code $\deg(g(X)) = n-k$ and for an MDS code $d=n-k+1$,
so an MDS BCH code necessarily verifies $\deg(g(X)) = d-1$). Seeing that
$g(X)$ already has $d-1$ roots over $\F_{q^m}$, it cannot have any other
roots. This means that the powers $\b^{\ell+j}$, $j=0,\dots,d-2$, must
all be conjugates of each other.

\subsubsection{The need for shortening.}

When building diffusion layers, the input and output of the diffusion
generally have the same size (otherwise inversion might be a problem),
so we need codes of length $2k$ and dimension $k$. In terms of BCH
codes, this translates into using $k$ consecutive powers of an element
$\b$ of order $2k$, and having $g(X)$ of degree $k$. Of course,
elements of even order do not exist in extensions of $\F_2$, so this
is not possible. Instead of using full length BCH codes, we thus use
\emph{shortened} BCH codes.

\begin{definition}[Shortened code]
  Given a $[n,k,d]_q$ code $\C$, and a set $I$ of $z$ indices
  $\{i_1,\dots,i_z\}$, the shortened code $\C_I$ of $C$ at indices from
  $I$ is the set of words from $\C$ which are zero at positions
  $i_1,\dots,i_z$, and whose zero coordinates are deleted, thus
  effectively shortening these words by $z$ positions. The shortened
  code $\C_I$ has length $n-\z$, dimension $\geq k-\z$ and minimal
  distance $\geq d$.

\end{definition}
  If $\C$ is MDS, then $d=n-k+1$ and $\C_I$ will necessarily be an $[n-\z, k-\z, d]_q$ MDS code, as neither the
  dimension nor the minimal distance can increase without breaking the Singleton bound~\cite{Singleton:IT64}.

  We can thus look for $[2k+\z, k+\z, k+1]_q$ BCH codes and shorten
  them on $\z$ positions to obtain our MDS codes. However, shortened
  BCH codes are no longer cyclic, so the shortening has to be done in a way that conserves
  the recursive structure. This is easy to achieve by using the
  previous systematic representation and shortening on the last
  positions. Starting from $g(X)$ of degree $k$, which divides
  $X^{2k+z}-1$, we get a generating matrix:
$$G = \arraycolsep=4pt \begin{pmatrix}X^k\bmod g(X)\smash{\hskip15pt\vrule width 1pt height 3mm depth 13.7mm\hskip-16pt}& 1 & 0 & 0 & 0\\
  X^{k+1}\bmod g(X)& 0 & 1 & 0 & 0\\
  \cdots &  & \hskip3pt\cdots\hskip-12pt\\
  X^{2k+\z-1}\bmod g(X)& 0 & 0 & 0 &
  1\end{pmatrix}.\hskip-2.1cm\lower.8cm\hbox{$\underbrace{\vrule height
    0pt depth 0pt width 1.7cm}_{\textrm{size } k+\z}$}$$ Shortening
the code on the $\z$ last positions will maintain the systematic form and
simply remove the $\z$ last lines to obtain:
$$G_I = \arraycolsep=4pt
\begin{pmatrix}
  X^k\bmod g(X)\smash{\hskip10pt\vrule width 1pt height 3mm depth 13.7mm\hskip-11pt}& 1 & 0 & 0 & 0\\
  X^{k+1}\bmod g(X)& 0 & 1 & 0 & 0\\
  \cdots &  & \hskip3pt\cdots\hskip-12pt\\
  X^{2k-1}\bmod g(X)& 0 & 0 & 0 & 1
\end{pmatrix}.\hskip-2.1cm\lower.8cm\hbox{$\underbrace{\vrule height
    0pt depth 0pt width 1.7cm}_{\textrm{size } k}$}$$
As said above, when $G$ generates an MDS code, then ${G_I}$ also
generates an MDS code, and this is (up to a permutation of the two $k\times k$ blocks, that will not affect the MDS property) exactly what we are
looking for: a recursive MDS matrix defined by the companion matrix
associated to the polynomial $g(X)$.

\subsection{Direct Construction of Recursive MDS Matrices}\label{sec.direct}

From this result, in the case where $q=2^s$, we can deduce a direct construction of recursive MDS matrices based on MDS BCH codes that were already described in~\cite{sloane}, Chapter~11, \S5. We first pick a $\b$ of order $q+1$. As
$q+1$ divides $q^2-1$, $\b$ is always in $\F_{q^2}$, the degree-2 extension of $\F_q$. Then, apart from $\b^0 = 1$, all
powers of $\b$ have minimal polynomials of degree 2: since $\b$ is of order $q+1$, each $\b^i$ has a conjugate $\b^{qi}=\b^{-i}$ which is the second root
of $\Min_{\F_q}(\b^i)$. From there, it is easy to build a $[q+1,q+1-k,k+1]_q$ MDS BCH code for any value of $k\leq\frac q2$.
\begin{itemize}
\item If $k$ is even, we need to select $k$ consecutive powers of $\b$ that are conjugates by pairs: if $\b^i$ is
  selected, $\b^{q+1-i}$ is selected too. We thus select all the powers $\b^i$ with $i\in [\frac {q-k}2+1,\frac
  {q+k}2]$, grouped around $\frac {q+1}2$.
\item If $k$ is odd, we need to select $\b^0$ as well. We thus select all the powers $\b^i$ with $i\in [-\frac
  {k-1}2,\frac {k-1}2]$, grouped around 0.
\end{itemize}
In both cases, we get a polynomial $g(X)$ of degree $k$ defining an MDS BCH code of length $q+1$. We can then shorten
this code on $z=(q+1-2k)$ positions and obtain the $[2k,k,k+1]_q$ MDS code we were looking for. The non-systematic part of
the generator matrix of this code is the $k$-th power of the companion matrix defined by $g(X)$.

Also, as the conjugate of $\b^i$ is its inverse, $g(X)$ enjoys the
same symmetry as the example of Section~\ref{sec.example}:
$X^kg(X^{-1}) = g(X)$. This explains the symmetry observed in~\cite{augot-finiasz-isit13}. Furthermore, the
companion matrix associated to $g(X)$ thus has at most $\frac k2$
different coefficients and can be implemented with at most $\frac k2$
multiplications.

Finally, by cycling over all $\b$ of order $q+1$, in the case where $2k=q$ we were able to recover with this direct construction all the solutions found
in~\cite{augot-finiasz-isit13} through exhaustive search. We conjecture that when $2k=q$, the only recursive MDS
matrices that exist come from these shortened BCH codes.

\section{An Algorithm to Find All MDS BCH Codes}\label{sec.algo}

We have seen that shortened BCH codes allow to directly build
recursive MDS matrices. However, when building a block cipher, the
designer usually has some parameters in mind (say, a diffusion layer
on $k$ symbols of $s$ bits each) and wants the \emph{best} diffusion
layer matching these parameters. Our direct construction gives good
solutions, but cannot guarantee they are the best. So the designer
needs an algorithm that will enumerate all possible matrices and let
him pick the most suitable one. For this, we will consider BCH codes
where $\b$ is a $(2k+\z)$-th root of unity and not only a $(2k+1)$-th
root of unity as in the direct construction. First, there are a few
constraints to consider.

\paragraph{Field Multiplication or $\F_2$-linearity?}

The designer has to choose the type of linearity he wants for his diffusion layer. If he wants (standard) linearity over
$\F_{2^s}$, then the BCH code has to be built over $\F_{2^s}$ (or a subfield of $\F_{2^s}$, but the construction is the
same). However, as in the Sajadieh~\etal~\cite{Sajadieh:FSE2012} or the Wu~\etal~\cite{WuWangWu:SAC2012} constructions,
he could choose to use an $\F_2$-linear operator $L$. Assuming $L$ has an irreducible minimal polynomial of degree
$s'\leq s$ (see~\cite{augot-finiasz-isit13} for how to deal with non-irreducible minimal polynomials), then he needs to build a BCH code
over $\F_{2^{s'}}$. This choice is up to the designer but does not change anything to the rest of the
algorithm, so we will assume $s' = s$.

\paragraph{The MDS Conjecture.}

Our shortened BCH construction starts by building an MDS code of length $2k+\z$ over $\F_{2^s}$. The MDS conjecture tells
us that $2k+\z\leq 2^s+1$ must hold. When $k=2^{s-1}$, $\z=1$ is the only choice. In general, we can choose any
$\z\in[1,2^s+1-2k]$, so the algorithm will need to try all these possibilities.

\paragraph{Minimal Polynomials of Roots of Unity.}

The $\b$ involved in the BCH construction is a $(2k+\z)$-th root of
unity, and $g(X)$ is formed as the product of minimal polynomials of
powers of $\b$. First, $(2k+\z)$-th roots of unity must exist, meaning
$2k+\z$ must be odd (or more generally coprime with $q$ when $q$ is
not $2^s$). Then, when factorizing $X^{2k+\z}-1$, the minimal
polynomials of the $\b^i$ are factors of this decomposition, and
$g(X)$ is the product of some of these factors. It must thus be
possible to obtain a polynomial of degree $k$ this way. This is not
always possible: for example, $X^{23}-1$ decomposes over $\F_{2^8}$ in
a factor of degree 1 and two factors of degree 11 and very few values of $k$
can be obtained. However, this last condition is rather complex to
integrate in an algorithm and it will be easier to simply not take it
into account.

\subsection{A Simple Algorithm}

For given parameters $k$ and $q=2^s$ we propose to use
Algorithm~\ref{alg.BCH} to enumerate all possible recursive MDS
matrices coming from shortened BCH codes. This algorithm explores all
code lengths from $2k+1$ to $q+1$, meaning that the number of
shortened columns can be much larger than the final code we are aiming
for.  Instead of computing minimal polynomials and their least common
multiple as in the definition of BCH codes we directly compute
$\prod_{j=0}^{k-1}(X-\b^{\ell+j})$ and check if it is in
$\F_q[X]$. This allows the algorithm to be more efficient and also
makes upper bounding its complexity much easier. The following lemma shows that the two formulations are equivalent.

\begin{lemma}
  A BCH code over $\F_q$ defined by the $d-1$ roots
  $[\b^\ell,...,\b^{\ell+d-2}]$ is MDS, if and only if $P(X) =
  \prod_{j=0}^{d-2}(X-\b^{\ell+j})$ is in $\F_q[X]$. In this case,
  $g(X) = \lcm\big(\Min_{\F_q}(\b^\ell),
  ... ,\Min_{\F_q}(\b^{\ell+d-2})\big)$ is equal to $P(X)$.
\end{lemma}

\begin{proof}
  We have seen that a BCH code is MDS if and only if $g(X)$ is of
  degree $d-1$ exactly. Also, $g(X)$ is always a multiple of $P(X)$.

  First, assume we have an MDS BCH code. Then $g(X)$ is of degree
  $d-1$ and is a multiple of $P(X)$ which is also of degree $d-1$. So,
  up to a scalar factor, $g(X) = P(X)$ and $P(X) \in \F_q[X]$.

  Conversly, assume we have a BCH code such that
  $P(X)\in\F_q[X]$. Then, for any $j\in[0,d-2]$, $P(X)$ is a
  polynomial in $\F_q[X]$ having $\b^{\ell+j}$ as a root, so $P(X)$ is
  a multiple of $\Min_{\F_q}(\b^{\ell+j})$. Therefore, $g(X)$ divides
  $P(X)$ and, as $P(X)$ also divides $g(X)$, we have $g(X) =
  P(X)$. $g(X)$ thus has degree $d-1$ and the code is MDS.\qed
\end{proof}

\begin{figure}[t]
\begin{algorithm}[H]
\caption{Search for Recursive MDS Matrices\label{alg.BCH}}
\DontPrintSemicolon
\KwIn{parameters $k$ and $s$}
\KwOut{a set $\mathcal S$ of polynomials yielding MDS matrices}
$q \leftarrow 2^s$\;
$\mathcal S \leftarrow \emptyset$ \;
\For{$\z \leftarrow 1$ \KwTo $(q+1-2k)$, \KwSty{with} $\z$ odd}{
  $\a \leftarrow$ primitive $(2k+\z)$-th root of unity of $\F_q$ \;
  \ForAll{$\b = \a^i$ \KwSty{such that}  $\ord(\b) = 2k+\z$}{
     \For{$\ell \leftarrow 0$ \KwTo $(2k+\z-2)$}{
       $g(X) \leftarrow \prod_{j=0}^{k-1}(X-\b^{\ell+j})$\label{line.gprod}\;
       \If(\hskip5mm\small \emph{(we test if $g$ has all its coefficients in $\F_q$)}){$g(X) \in \F_q[X]$}{
         $\mathcal S \leftarrow \mathcal S \cup \{g(X)\}$\;
       }
     }
  }
}
\Return{$\mathcal S$}
\end{algorithm}
\caption{Algorithm searching for MDS BCH codes}\label{fig:simple_algo}
\end{figure}

\subsection{Complexity}

The previous algorithm simply tests all possible candidates without trying to be smart about which could be eliminated
faster. It also finds each solution several times (typically for $\b$ and $\b^{-1}$), and finds some \emph{equivalent}
solutions (applying $\a \mapsto \a^2$ on all coefficients of the polynomial preserves the MDS property, so each
equivalence class is found $s$ times).

The product at line~\ref{line.gprod} does not have to be fully recomputed for each value of $\ell$. It can be computed
once for $\ell =0$, then one division by $(X-\b^\ell)$ and one multiplication by $(X-\b^{\ell+k})$ are enough to update
it at each iteration. This update costs $O(k)$ operations in the extension of $\F_q$ containing $\a$. The whole loop
on $\ell$ can thus be executed in $O((2k+\z)k)$ operations in the extension field.

The number of $\b$ for which the loop has to be done is Euler's phi
function $\varphi(2k+\z)$ which is smaller than $(2k+\z)$, itself
smaller than $q$, and there are $\frac{q-2k}2+1$ values of $\z$ to
test. This gives an overall complexity of $O(q^2k(q-2k))$ operations
in an extension of $\F_q$. This extension is of degree at most
$2k+\z$, so operations are at most on $q\log q$ bits in this extension
and cost at most $O(q^2(\log q)^2)$. This gives an upper bound on the
total complexity of $O\big(q^4k(q-2k)(\log q)^2\big)$ binary
operations, a degree-6 polynomial in $k$ and $q$. This is a quite
expensive, but as we will see in the next section, this algorithms
runs fast enough for most practical parameters. It should also be
possible to accelerate this algorithm using more elaborate computer
algebra techniques.

\section{Experimental Results}

We implemented Algorithm~\ref{alg.BCH} in Magma~\cite{magma} (see the code in Appendix~\ref{sec.magma}) and ran it for
various parameters.

\subsection{The Extremal Case: $2k=2^s$.}

First, we ran the algorithm for parameters on the bound given by the
MDS conjecture, that is, when $2k = 2^s$. These are the parameters
that were studied by Augot and Finiasz
in~\cite{augot-finiasz-isit13}. It took their algorithm 80 days of CPU
time to perform the exhaustive search with parameters $k=16$ and $s=5$
and find the same 10 solutions that our new algorithm finds in a few
milliseconds. The timings and number of solutions we obtained are
summarized in Table~\ref{tab.borneMDS}. We were also able to find much
larger MDS diffusion layers. For example, we could deal with $k=128$
elements of $s=8$ bits, which maybe is probably too large to be
practical, even with a recursive structure and the nice
symmetry. Below are the logs in base $\a$ (with
$\a^8+\a^4+\a^3+\a^2+1=0$) of the last line of the companion matrix of
an example of such 1024-bit diffusion:
\begin{multline*}
  [0, {83}, {25}, {136}, {62}, {8}, {73}, {112}, {253}, {110}, {246}, {156}, {53},
  {1}, {41}, {73}, {5}, {93}, {190}, {253}, {149},\\ {98}, {125}, {124}, {149}, {94},
  {100}, {41}, {37}, {183}, {81}, {6}, {242}, {74}, {252}, {104}, {57}, {117},
  {55}, {224},\\ {153}, {130}, {77}, {156}, {192}, {176}, {52}, {133}, {218}, {59},
  {158}, {18}, {228}, {89}, {218}, {126}, {146},\\ {210}, {217}, {18}, {84}, {209},
  {30}, {123}, {\boldsymbol{{97}}}, {123},  \dots \textrm{\small [ symmetric ]}\dots  , {83} ]
\end{multline*}

\begin{table}[t]
  \centering
  \caption{Experimental results for parameters on the bound given by MDS conjecture. The value ``diff. bits'' is the size in bits of the corresponding diffusion layer. The number of solutions is given as both the raw number and the number of distinct equivalence classes.\label{tab.borneMDS}}
  \tabcolsep=4pt
  \begin{tabular}{|c|c|c|c|c|c|}\hline
    \smash{\lower6pt\hbox{$k$}} & \smash{\lower6pt\hbox{$s$}} & diff. & \multicolumn{2}{c|}{solutions} & \smash{\lower6pt\hbox{time}} \\
        &     & bits  & num. & classes &                      \\\hline
    4   & 3   & 12    & 3    & 1       & $<$0.01s \\
    8   & 4   & 32    & 8    & 2       & $<$0.01s \\
    16  & 5   & 80    & 10   & 2       & $<$0.01s \\
    32  & 6   & 192   & 24   & 4       & $\sim$0.02s \\
    64  & 7   & 448   & 42   & 6       & $\sim$0.07s \\
    128 & 8   & 1024  & 128  & 16      & $\sim$0.52s \\
    256 & 9   & 2304  & 162  & 18      & $\sim$1.71s \\\hline
  \end{tabular}
\end{table}

\begin{table}[t]
  \centering
  \caption{Experimental results for other interesting parameters. The reg. solutions refer to regular solutions where the constant term of the polynomial is~1.\label{tab.gener}}
  \tabcolsep=4pt
  \begin{tabular}{|c|c|c|c|c|c|}\hline
    \smash{\lower6pt\hbox{$k$}} & \smash{\lower6pt\hbox{$s$}} &
diff. & \multicolumn{2}{c|}{solutions} & \smash{\lower6pt\hbox{time}}
\\ & & bits & num.  & reg. & \\\hline 4 & 4 & 16 & 68 & 12 &
$\sim$0.02s \\ 4 & 8 & 32 & 20180 & 252 & $\sim$37s \\ 8 & 8 & 64 &
20120 & 248 & $\sim$44s \\ 16 & 8 & 128 & 19984 & 240 & $\sim$55s \\
32 & 8 & 256 & 19168 & 224 & $\sim$80s \\\hline

  \end{tabular}
\end{table}

\subsection{The General Case} We also ran some computations for other
interesting parameters, typically for values of $k$ and $s$ that are
both powers of 2 as it is often the case in block ciphers. The results
we obtained are summarized in Table~\ref{tab.gener}. Note that for these
solutions the number of
shortened positions is sometime huge: for $k=4$ and $s=8$ one can
start from a $[257,253,5]_{256}$ BCH code and shorten it on 249 positions to obtain a
$[8,4,5]_{256}$ code. We counted both
the total number of solutions we found and the number of regular
solutions where the constant term of the polynomial is~1. Regular
solutions are particularly interesting as the diffusion and its
inverse share the same coefficients:
\begin{equation}\label{eq.regular_inverse} \comp(1,c_1,\dots,c_{k-1})^{-1}
=\arraycolsep=3pt \begin{pmatrix}0 & 1 & & 0\\ &
&\hskip-1mm\smash{\ddots}\hskip-1mm & \\ 0 & 0 & & 1\\1& c_1 & &
c_{k-1}\end{pmatrix}^{-1}\hskip-3mm = \begin{pmatrix}c_{1} & & c_{k-1} & 1\\1& &
0 & 0\\ & \hskip-1mm\smash{\ddots}\hskip-1mm & & \\0 & & 1 &
0\end{pmatrix}.
\end{equation} In the case of symmetric solutions (like those from
Section~\ref{sec.direct}), encryption and decryption can even use the
exact same circuit by simply reversing the order of the input and
output symbols. Here are some examples of what we found:
\begin{itemize}
\item for parameters $k=4$ and $s=4$, with $\a$ such that $\a^4+\a+1 = 0$, the matrices $\comp(1, \a^3, \a, \a^3)^4$ and $\comp(\a^3+\a,1,\a,\a^3)^4$ are MDS.
\item for parameters $k=4$ and $s=8$, with $\a$ such that $\a^8+\a^4+\a^3+\a^2+1 = 0$, the matrices $\comp(1, \a^3, \a^{-1}, \a^3)^4$, $\comp(1, \a^3+\a^2, \a^3, \a^3+\a^2)^4$, and $\comp(\a+1, 1, \a^{202} + 1, \a^{202})^4$ are MDS.
\end{itemize}

The reader might note the absence of larger fields in
Table~\ref{tab.gener}. One could for example want to obtain a 128-bit
diffusion layer using $k=8$ symbols of $s=16$ bits. However, going
through all the possible values of $\z$ and $\ell$ takes too long with
$q=2^{16}$. Our algorithm is too naive, and an algorithm enumerating
the divisors of $X^{2k+\z}-1$ of degree $k$ and checking if they
correspond to BCH codes could be faster in this case. Otherwise, it is
always possible to use the direct construction given in
Section~\ref{sec.direct}.

\subsection{Further Work}

As we have seen, for most parameters, this algorithm runs fast enough to find all recursive MDS matrices coming from BCH
codes. However, not all recursive MDS matrices come from a BCH code.
\begin{itemize}
\item First, there are other classes of cyclic codes that are MDS and could be shortened in a similar way. Any such class of
codes can directly be plugged into our algorithm, searching for polynomials $g(X)$ having another structure than roots
that are consecutive powers of $\b$.
\item Then, there also are cyclic codes which are not MDS, but become MDS once they are shortened. These will be much harder
to track as they do not have to obey the MDS conjecture and can have a much larger length before shortening.
\end{itemize}
For this reason, we are not always able (yet) to find the most
efficient matrices with our algorithm. For example, the matrix used in
Photon corresponds to a cyclic code of length $2^{24}-1$ over
$\F_{2^8}$ which is not MDS. We know that this code has minimum
distance 3, and its distance grows to 5 when shortened from the length
$2^{24}-1$ to the length 8.

However, for some parameters, our algorithm is able to find very nice
solutions. For $k=4$ and $\a$ verifying $\a^5+\a^2+1 = 0$ (a primitive
element of $\F_{2^5}$, or an $\F_2$-linear operator with this minimal
polynomial), the matrix $\comp(1,\a,\a^{-1},\a)$ found by
Algorithm~\ref{alg.BCH} yields an MDS diffusion layer. This is
especially nice because it is possible to build simple $\F_2$-linear
operators that also have a simple inverse, and this solution is
symmetric meaning the inverse diffusion can use the same circuit as
the diffusion itself.

\section{Conclusion}

The main result of this article is the understanding that recursive
MDS matrices can be obtained directly from shortened MDS cyclic
codes. From this, we derive both a direct construction and a very
simple algorithm, based on the enumeration of BCH codes, that allows
to efficiently find recursive MDS matrices for any diffusion and
symbol sizes. These constructions do not always find all existing
recursive MDS matrices and can thus miss some interesting
solutions. As part of our future works, we will continue to
investigate this problem, trying to understand what properties the
other solutions have and how we can extend our algorithm to find them
all. A first step is to elucidate the Photon matrix in terms of cyclic
codes which are not BCH codes, hopefully finding a direct construction
of this matrix. However, in the same way as computing the minimal
distance of a cyclic code is difficult, it might turn out that finding
all recursive MDS matrices of a given size is a hard problem.

\bibliographystyle{plain}

\appendix

\section{Magma Code}\label{sec.magma}

Here is the Magma code for Algorithm~\ref{alg.BCH}. Simply run \texttt{BCH(k,s)} to get the set of all polynomials of
degree $k$ over $\F_{2^s}$ that yield MDS diffusion layers on $ks$ bits. Of course, these polynomials have to be written
as companion matrices which then have to be raised to the power $k$ to obtain the final MDS matrices.

\begin{Verbatim}
BCH := function(k,s)
  q := 2^s;
  F := GF(q);
  P := PolynomialRing(F);
  S := { };
  for z:=1 to q+1-2*k by 2 do
    a := RootOfUnity(2*k+z, F);
    Pext<X> := PolynomialRing(Parent(a));
    for i:=0 to 2*k+z-1 do
      b := a^i;
      if Order(b) eq (2*k+z) then
        g := &*[(X-b^l): l in [-1..k-2]];
        for l in [0..2*k+z-2] do
          g := (g*(X-b^(l+k-1))) div (X-b^(l-1));
          if IsCoercible(P,g) then
            Include(~S, P!g);
          end if;
        end for;
      end if;
    end for;
  end for;
  return S;
end function;
\end{Verbatim}

\end{document}